\RequirePackage{fix-cm}
\documentclass[smallextended]{svjour3}       % onecolumn (second format)
\smartqed  % flush right qed marks, e.g. at end of proof
\usepackage{graphicx}

\usepackage{amsmath}%,amsthm,amssymb}
\usepackage{amssymb}
\usepackage{xcolor}
\usepackage{listings}
\usepackage{relsize}
\usepackage{enumitem}

\usepackage{algorithm}
\usepackage[noend]{algpseudocode}

\begin{document}

\title{Large Datasets, Bias and Model Oriented Optimal Design of Experiments 
%\thanks{}
}
% Grants or other notes about the article that should go on the front
% page should be placed within the \thanks{} command in the title
% (and the %-sign in front of \thanks{} should be deleted) 
%
% General acknowledgments should be placed at the end of the article.

% \subtitle{Do you have a subtitle?\\ If so, write it here}


\author{Elena Pesce         \and
        Eva Riccomagno %etc.
}


\institute{E. Pesce \at
              Department of Mathematics, Genoa (Italy) \\
              \email{pesce@dima.unige.it}           %  \\
           \and
           E. Riccomagno \at
              Department of Mathematics, Genoa (Italy)\\
              \email{riccomagno@dima.unige.it}
}

\date{Received: date / Accepted: date}

\maketitle

\begin{abstract}
We review recent literature that proposes to adapt ideas from classical
model based optimal design of experiments to problems of data
selection of large datasets. Special attention is given to bias
reduction and to protection against confounders. Some new results are presented. 
Theoretical and computational comparisons are made.

\end{abstract}

\section{Introduction}
\label{intro}

For the analysis of big datasets statistical methods have been developed which use the full available dataset.  
For example new methodologies developed in the context of big data and focussed on a ``divide-and-recombine" approach 
are summarised in~\cite{wang2015statistical}. 
Other two major methods address the scalability of big data through Bayesian inference based on a Consensus Monte Carlo algorithm~\cite{scott2016bayes} %(Scott et al., 2016) 
and sparsity assumptions~\cite{tibshirani2015statistical}.% (Tibshirani et al., 2015). 

In contrast other authors argue on the advantages of inference statements based on a well-chosen subset of the big dataset. 
Below we review some algorithms and papers for the model based selection of subsamples from a large dataset. 
While usually data can be collected in scientific studies via active or passive observation, big data is often collected in passive way. 
Rarely their collection is the result of a designed process. 
This generates sources of bias which either we do not know at all or are too costly to control.
Nevertheless they will affect the overall distribution of the observed variables~\cite{dunson2018,pescecladag}.

Many authors in~\cite{specialIssueStatProbLetters} argues that analysis of big data set is effected by issues of bias and confounding, selection bias and other sampling problems  (e.g.~\cite{sharpes2018} for electronic health records).  
Often the causal effect of interest can only be measured on the average and great care has to be 
taken about the background population. 
The analysis of the full dataset might be prohibitive because of computational and time
constraints. 
Indeed in some cases the analysis of the full dataset might also be not advisable~\cite{Harford2014}. 
To recall just one example,~\cite{Meng2018} reports that the simple sample proportion of a self-reported big dataset of size $2,300,000$ unit has the same mean squared error as the  sample proportion from a suitable simple random sample of size $400$ and a Law of Large Population has been defined in order to qualify this.  
% It also proves a decomposition 
%of 
%the difference between the sample average  and the population average
% as the product of three terms and relates them to different sources of bias. 

%When we talk about ``Big Data'' we refer to data on a massive scale whose size exceed the capacity of a single computer~\cite{wang2015statistical}. This definition is widely accepted, but in this context it may seem too general because it is not possible to distinguish the cases in which 
%we have many observations $n$ but few predictors $p$ and the opposite. Hence we prefer to refer to tall dataset when $n >> p$ and to wide dataset when $p >> n$, while a large dataset indicates a general dataset with $n \approx p$.

Recently some researchers argued on the usefulness of utilising methods and ideas from  Design of Experiment (DoE) for the analysis of big datasets,  more specifically from model-based optimal experimental design. 
They argue that special models are useful, or even needed, to guard against hidden sources of bias and  that a well-chosen subset of the big dataset can deliver equivalent answers compared to the full dataset at considerably less effort.  
For example one can resort to using randomization or latent variable methods. 
In Section~\ref{section_moda} we review some of those papers 
%that propose to adapt ideas from classical model based optimal DoE to problems of data selection from big datasets
 (see also~\cite{flassig2018model}) distinguishing models without bias, models with bias and no confounders, and models with confounders and no bias.
We make some steps towards the generalisation to include both bias and confounders in Section~\ref{general_formulation}. 
  % Flassig and Schenkendorf, 2018). 
%First we consider classical linear models, secondly models with a bias term and then models with a terms modelling confounders. 
%We  refer on those papers which do not specifically consider the bias term in Subsection~\ref{section_modelwithoutbias}
% and in Subsection~\ref{section_modelwithtbias} we focus on the use of methods from optimal DoE for bias reduction
% and in Subsection~\ref{section_onlyconfounders} on those with only confounders.  
% Some theoretical and computational comparisons are made among these approaches using the software R.
Theoretical and computational comparisons made using the software R lead us to conclude that so far these approaches are more suitable for tall dataset than for genuine large datasets and indicate that much work is needed to have efficient algorithms for subsample selection from large datasets in the presence of bias and confounders.
To fix terminology we recall that a dataset is tall if the number of observations is much larger than the number of predictors, and large when it has many observations and predictors. 
%Finally Section~\ref{general_formulation} presents some results towards the study of models with both bias and confounder terms.
%In Section~\ref{general_formulation} we present some first results for when both bias and confounder terms are present in the model. 

\section{Model oriented selection of sub-dataset}
\label{section_moda}

%Design of Experiments model-based methods have been recently developed following the idea that a well-chosen subset of the big dataset can deliver equivalent answers compared to the full dataset at considerably less effort.  

The most general form of the considered model  is that of a linear model for a response variable $\pmb{Y}$ 
\begin{equation}
\label{model}
\pmb{Y}(\pmb{x},\pmb{z}) = \pmb{f}^\prime(\pmb{x}) \pmb{\theta} + \pmb{h}^\prime(\pmb{x}) \pmb{\psi} + \pmb{g}^\prime (\pmb{z}) \pmb{\phi} + \epsilon
\end{equation}
with $\pmb{\theta} \in \mathbb{R}^p$, $\pmb{\psi} \in \mathbb{R}^m$ and $\pmb{\phi} \in \mathbb{R}^q$ and 
with $\pmb{x} \in \mathcal{X}$, $\pmb{z} \in \mathcal{Z}$. 
The observed values are on the $\pmb{x}$, while the $\pmb{z}$ are unknown. 
Both the $ \mathcal{X}$ and $ \mathcal{Z}$ spaces are assumed to be finite and $^\prime$ indicates transpose.
The usual assumptions are taken on the random errors: $\epsilon_i$ are iid and $\operatorname{Var}(\epsilon_i) = \sigma^2$. 
There are three terms in the model: the first corresponds to a classical linear  model, the second to a bias term related to the variables $\pmb{x}$ and the last term models a bias that may result from confounders, sources of bias which either we do not know at all or are too costly to control. We assume it to be linear for simplicity of comparison. 

Special cases of the Model in Equation~(\ref{model}) have been addressed in order to adapt ideas from classical model based optimal DoE: \cite{montepiedra} and \cite{wiens} consider $\pmb{f}^\prime(\pmb{x}) \pmb{\theta} + \pmb{h}^\prime(\pmb{x}) \pmb{\psi} $,
while \cite{pescecladag} considers
$\pmb{f}^\prime(\pmb{x}) \pmb{\theta} + \pmb{g}^\prime (\pmb{z}) \pmb{\phi} $.
All search for a design which minimises the mean square error of the least square estimate (LSE) of the $\pmb{\theta}$ parameters, guarding against the two different sources of bias. 
Recently, authors of~\cite{drovandi} and~\cite{stufken} proposed methods of data selection from large datasets in a DoE context, as a response to the more and more frequent need to analyse Big Data. However they do not guard against different sources of bias. We review these first.

\subsection{Model without bias} 
\label{section_modelwithoutbias} 

In this section we consider the model $  \operatorname{E}\left[ Y(\pmb{x}) \right] =\pmb{f}( \pmb{x})^\prime \, \pmb{\theta}  $ and the two algorithms presented in~\cite{drovandi} and~\cite{stufken}. 
An optimal experimental design perspective is suggested in~\cite{drovandi}, where a retrospective sample set is drawn in accordance with a sampling plan or experimental design. Analysis and inference are then based on this designed sample. This approach is targeted towards applications of regression models with large number of observations and relative small number of predictors, otherwise the problem of finding the best subset of data becomes computationally hard or infeasible due to the curse of dimensionality. 

The pseudocode of the algorithm is presented in Algorithm~\ref{Drovandi}. 
The input to the algorithm is the support vector of the mean of a linear model,
$  \pmb{f}   $, 
a utility function $U$ based on $ \pmb{f}$, a distance function in $\mathcal X$  
and  a tall dataset called  {\bf Data}
 with typical row $(\pmb{x} , \pmb{y}(\pmb{x}))$ with  $\pmb{x} \in  {\bf Data} \subset \mathcal X$. 
 In~\cite{drovandi} various $U$ functions are considered and the Euclidean distance.   
The output of the algorithm is a subset of 
 $n_d$ data points from {\bf Data} which maximises some expected utility $U$,  where $n_d$ is much smaller than the number of points in $\mathcal X$. 
 
The key idea behind the algorithm is to ``cluster'' $\mathcal X$, or to ``discretise'' it, into a grid.
Then $ \pmb{\theta}   $ is estimated using an initial random sample from $\mathcal X$. 
Next  the grid point $\pmb{d^\ast}$ maximising $U$ is found and 
one or more $\pmb{x}$ points in \textbf{Data} which are closest to  $\pmb{d^\ast}$ with respect to specified distance are added to the random sample. 
This is repeated until a subset of size $n_d$  is obtained. 
 
\begin{algorithm}
\caption{Sample selection from  {\bf Data} based on $  \operatorname{E}\left[ Y(\pmb{x}) \right] {=}\pmb{f}( \pmb{x})^\prime \, \pmb{\theta}   $ according to~\cite{drovandi} }\label{Drovandi}
\begin{algorithmic}[1] 
\State Fix a grid on $\mathcal X$
\State Sample randomly a subset of size $n_t < n_d$ from the  {\bf Data} and obtain the estimate of ${\pmb{\theta}}$ or form a prior density function $p(\pmb{\theta})$. Set the current sample size $n_c = n_t$
\While  {$n_c \le n_d$ or when a certain criteria is not met}
\State Find the grid point $\pmb{d^\ast}$ such that $\pmb{d^\ast} {=} 
  \underset{\pmb{d} }{\operatorname{argmax} } \, 
 \mathbb{E}[U(\pmb{d}, \pmb{\theta}, {\pmb{y}}(\pmb{d}))]$
\State Find $\pmb{x}$ in \textbf{Data} and not already sampled, which minimizes the distance $ || \pmb{x} - \pmb{d}^\ast || $ \State Add $( \pmb{x}, \pmb{y}(\pmb{x}) )$ into the data subset, remove the observation  $( \pmb{x}, \pmb{y} (\pmb{x}) )$ from \textbf{Data} and set $n_c \leftarrow n_c + 1$ 
(steps 5 and 6 may be performed multiple times to sub-sample a batch of data of size $m$, and setting $n_c \leftarrow n_c + m$)
\State Re-estimate ${\pmb{\theta}}$ or update the prior distribution $p(\pmb{\theta})$
\State \textbf{go to while}
\EndWhile
\end{algorithmic}
\end{algorithm}
 
The major features of Algorithm~\ref{Drovandi} is that it returns a subset of {\bf Data} via an optimal, sequential and response adaptive procedure. 
The computations of the distances in point 5. and the optimisation problem in point 4.  
can be parallelised, thus speeding it up considerably. 
Parallelization is particularly useful when the stopping criterion, the utility function and/or the distance function are costly to evaluate or
when the sampling grid is large. 
Its major drawback is that it requires full trust in the model.
Also although it  can be adapted for variable selection, the algorithm  is efficient only for tall datasets, indeed point 5. and the computation of $U$ may suffer from  the curse of dimensionality.
Finally we note that 
the obtained optimal design  can be used as train set of, e.g., a random forest, giving interesting results (see Example~\ref{mortgage}). 

%%%%%%%%%
%%%%%%%%%
%%%%%%%%%
%%%%%%%%%
%%%%%%%%%

The second algorithm we present appears in~\cite{stufken} and 
is called IBOSS (Information-Based Optimal Subdata Selection).  
It is a deterministic algorithm to select the most informative data points 
for 
 the model $  \operatorname{E}\left[ Y(\pmb{x}) \right] =\pmb{f}( \pmb{x})^\prime \, \pmb{\theta}   $.
A pseudocode is given in Algorithm~\ref{Stufken}. 
The rational behind IBOSS is that $D$-optimal designs tend to be on the boundary of the available space.
 The selected points are shown to be optimal in the following sense.

 Let {\bf Data} have $N$ points. {\bf Data} can coincide with $\mathcal X$. 
A subset of  {\bf Data} of size $n_d$ is sought which maximises a univariate optimality criterion function $\Psi$ of the information matrix 
\[
\pmb{M}(\pmb{\delta}) = \frac{1}{\sigma^2} \sum_{i = 1}^{N} \delta_i \pmb{x}_i \pmb{x}_i' \, \, \text{ subject to  } \sum_{i = 1}^{N} \delta_i = n_d
\]
where 
$\delta_i = 1$  if point $i \in $ \textbf{Data} is selected and $\delta_i = 0$ otherwise.
%\[
%\delta_i = \begin{cases}
%1 & \text{if point $i \in $ \textbf{Data} is selected}\\
%0 & \text{otherwise}
%\end{cases}
%\]
The function $\Psi$ could be the determinant of $\pmb{M}$, expressing thus $D$-optimality.
In~\cite{stufken} the following inequality is proven for the $D$-optimality criterion
\[
\operatorname{det} (\pmb{M}(\pmb{\delta}) ) \le 4 \left( \frac{n_d}{4 \sigma^2} \right)^{p+1} \prod_{j=1}^{p} (x_{(N)j} - x_{(1)j})^2
\]
where $x_{(N)j} - x_{(1)j}$ is the observed range of the $j$th variable and $\sigma^2$ is the model variance. 
This gives a function easy to optimize and larger than the desired utility function. 
Thus the optimal design is obtained by selecting iteratively $r { = } n_d/(2p)$ data points on the boundary of the observed range of each predictor. 
If $\frac{n_d}{2p}$ is not integer, one can clearly take floor or ceiling or choose a suitable $n_d$.
 Note that  the full sample $\mathcal{X}$ does not need to be specified nor it is used.
 But the representativeness of {\bf Data} for $Y$ in $\mathcal X$ has to be trusted.

We found the algorithm to work better for tall datasets and tested it for up to one million points in four variables (much larger datasets are considered in~\cite{stufken}). 
It proved to be cost effective and can be parallelised.  
It requires full trust in the model and the output depends on the initial ordering of the variables in Step 3,  
as shown in the small two dimensional example in Figure~\ref{toy} where the designs obtained starting with the variable $x_1$ (in red) or $x_2$ (in green) can be very different. 
For special cases a symmetry argument or a group action could be employed to establish the equivalence of the obtained design. 

\begin{algorithm}
\caption{Pseudocode for IBOSS~\cite{stufken} %Sample selection from {\bf Data} based on $  \operatorname{E}\left[ Y(\pmb{x}) \right] {=} \pmb{f}( \pmb{x})^\prime \, \pmb{\theta}   $ according to~\cite{stufken}
}
\label{Stufken}
\begin{algorithmic}[1]
\State Assume $r = n_d/(2p)$ integer %Use a partition-based selection algorithm to perform each step
\State Initialise selected sample $\pmb{\delta}^\ast = \emptyset$
\For  {$j = 1, \ldots, p$}
\For {$1\le i \le N$}
\If {$j > 2$}
\State $\textbf{Data} \longleftarrow \textbf{Data} \setminus \pmb{\delta}^\ast$
\EndIf 
\State Add the $r$ points with the smallest $x_{ij}$ value to $\pmb{\delta}^\ast$
\State Add the $r$ points with the largest $x_{ij}$ values to $\pmb{\delta}^\ast$
\EndFor
\State\textbf{end}
\EndFor 
\State\textbf{end}
%\State $\hat{\pmb{\beta}}^D = \left\{ (\pmb{X}_D^\star)' \pmb{X}_D^\star \right\}^{-1}(\pmb{X}_D^\star)' \pmb{y}_D^\star $ and the estimated covariance matrix $\hat{\sigma}^2_D \left\{ (\pmb{X}_D^\star)' \pmb{X}_D^\star \right\}^{-1}$, where $\pmb{X}_D^\star = (\pmb{1}, \pmb{Z}_D^\star)$, $\pmb{Z}_D^\star$ is the covariance matrix of the subdata selected in the previous steps, $\pmb{y}_D^\star$ is the response vector of the subdata and $\hat{\sigma}^2_D = \left|\left| \pmb{y}_D^\star - \pmb{X}_D^\star \hat{\pmb{\beta}}^D\right|\right|^2 / (k-p-1)$
\State Obtain an estimate of  $\pmb{\theta}$ with the selected $n_d$ data points
\end{algorithmic}
\end{algorithm}

\begin{figure} 
\includegraphics[width=0.5\textwidth]{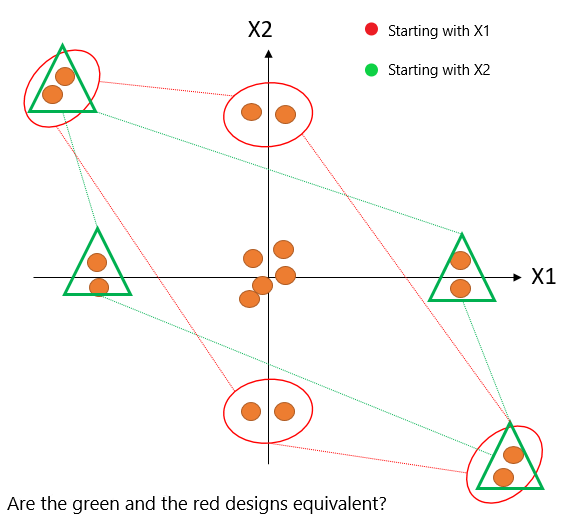}
\caption{Effect of the initial ordering of the factors on the IBOSS output}
\label{toy}
\end{figure} 

%\vspace{0.2cm}
%
%\begin{itemize}
%\item \underline{Special cases}:\vspace{0.01cm}
%\begin{enumerate}
%\item $\pmb{Y}(\pmb{x}) = \pmb{f}(\pmb{x})^\prime \pmb{\theta} + \pmb{\epsilon}$ \hfill [1,5]
%\item $\pmb{Y}(\pmb{x}) = \pmb{f}(\pmb{x})^\prime \pmb{\theta} + \pmb{h}(\pmb{x})^\prime \pmb{\psi} + \pmb{\epsilon}$  \hfill [3,6]
%\item $\pmb{Y}(\pmb{x, z}) = \pmb{f}(\pmb{x})^\prime \pmb{\theta} + \pmb{g}(\pmb{z})^\prime \pmb{\phi} + \pmb{\epsilon}$ \hfill [4]
%\end{enumerate}
%\end{itemize}
%

\subsection{Models with no confounders terms} 
\label{section_modelwithtbias} 

The model of the form 
$\pmb{Y}(\pmb{x}) = \pmb{f}(\pmb{x})^\prime \pmb{\theta} + \pmb{h}(\pmb{x})^\prime \pmb{\psi} + \pmb{\epsilon}$ 
is considered in~\cite{wiens}. 
The sample space  $\mathcal X$ is assumed to be a discrete finite set
$\mathcal X= \left \{
\pmb{x_1},   \ldots, \pmb{x_N} 
 \right \}$
  and $\mathcal X$ can be thought of  as the grid discretising the sample space in~\cite{drovandi}.

Several methods are presented in~\cite{wiens}  for the construction of designs that are minimax robust for linear or nonlinear models whose mean structures cannot be guaranteed to have been specified with complete accuracy. (Actually the author considers a more general bias term than the one above, specifically the model 
 $  \operatorname{E}\left[ Y(\pmb{x}) \right] =\pmb{f}( \pmb{x})^\prime \, \pmb{\theta}      +\psi(\pmb{x})$ under the constraint
$ \sum_{\pmb{x}\in \mathcal X} \pmb{f}(\pmb{x}) \psi(\pmb{x}) = 0 $ which ensures identifiability.)  
The classical notions of $I$- and $D$-optimality are extended by taking into account the bias of the predictions and a minimax $I$- and $D$-robust design theory is developed. 
Imposing a neighbourhood structure on the regression response function, the proposed methods maximise the mean squared error over this neighbourhood, and then seek $I$- and $D$- robust designs that minimize this maximum loss. 
% He also proposes a sequential method of design construction, in which each successive design point is chosen so as to minimize the subsequent loss.

Formally, 
let  $\hat{\pmb{\theta}}$ be the LSE of $\pmb{\theta}$ based on a design $\pmb{\xi}$.
The two  loss functions
\[
\mathcal{I}(\psi, \pmb{\xi}) {=} \sum_{\pmb{x} \in \chi} \mathbb{E} \left[ \left( \pmb{f}'(\pmb{x})\hat{\pmb{\theta}} - \mathbb{E}[Y(\pmb{x})]  \right)^2 \right] 
\quad \text{and} \quad
\mathcal{D}(\psi,\pmb{\xi}) {=} \left( \operatorname{det} \mathbb{E} \left[ (\hat{\pmb{\theta}} - \pmb{\theta} ) (\hat{\pmb{\theta}} - \pmb{\theta} )' \right]  \right)^{1/p}
\]
%is developed in~\cite{wiens} 
%where  $\hat{\pmb{\theta}}$ is the LSE of $\pmb{\theta}$ based on a design $\pmb{\xi}$. %  (subset of $\mathcal X$).
%
%The loss functions 
can be factorised as 
\[
\max_{\psi} \mathcal{I}(\psi,\pmb{\xi}) {=} \frac{\sigma^2+\tau^2}{n} \times \mathcal{I}_\nu (\pmb{\xi})
\quad \text{and} \quad
\max_{\psi} \mathcal{D}(\psi,\pmb{\xi}) {=} \frac{\sigma^2}{n} \left( \frac{\sigma^2+\tau^2}{\sigma^2 \operatorname{det}[\pmb{F}'\pmb{F}]} \right)^{1/p} \times \mathcal{D}_\nu(\pmb{\xi})
\]
where $\mathcal{I}_\nu$ and $\mathcal{D}_\nu$ depend only on the sought design and on known quantities.  
Here $\pmb{F} =  \left[ \pmb{f}(\pmb{x})\right]_{\pmb{x}\in \mathcal X}$ is the full model matrix. 
The objective is  to find $\underset{\pmb{\xi}}{\operatorname{min} }  \,\,  \underset{\psi}{ \operatorname{max} }\,\,  \mathcal{D}(\psi,\pmb{\xi}) $
 or $\underset{\pmb{\xi}}{\operatorname{min} }  \,\,  \underset{\psi}{ \operatorname{max} }\,\,  \mathcal{I}(\psi,\pmb{\xi}) $.

In more details 
$\sigma^2$ is the model variance,  $\tau$ a control parameter and $n$ the number of points in the design $\pmb{\xi}$ with non zero probability mass.
The design measure is indicated with $ \pmb{\xi} = \left\{
\begin{array}{l} \pmb{x} \\ \pmb{\xi}_{\pmb{x}} 
\end{array}\right\}_{  \pmb{x} \in \mathcal X}
 $   and can be collected in a diagonal matrix $\pmb{D}(\pmb{\xi})=\operatorname{diag}\left( \pmb{\xi}_{\pmb{x}} : \pmb{x} \in \mathcal X  \right)$.
The key factors in $\mathcal{I}(\psi,\pmb{\xi}) $ and $ \mathcal{D}(\psi,\pmb{\xi}) $ are
\begin{align*}
\mathcal{I}_\nu(\pmb{\xi}) &= (1-\nu) \operatorname{ tr } \pmb{R}^{-1}(\pmb{\xi}) + \nu \lambda_{max} \left(  \pmb{U}(\pmb{\xi}) \right)\\
\mathcal{D}_\nu(\pmb{\xi}) &= \left(  \frac{1-\nu+\nu \lambda_{max} \left(
  \pmb{R}^{1/2}(\pmb{\xi})[\pmb{U}(\pmb{\xi})-\pmb{I}_p]R^{1/2}(\pmb{\xi})
  \right) 
  }
  {\det[\pmb{R}(\pmb{\xi})]}   \right)^{1/p}
\end{align*}
where  
 $\nu = \tau^2/(\sigma^2 + \tau^2)$, $\pmb{I}_p$ the $p\times p$ identity matrix, $\lambda_{max} $ is the maximum eigenvalue.
 The control parameter $\nu$ is in $ [0,1]$. For $\nu=0$, then 
  $\mathcal{D}_0(\pmb{\xi}) $ gives the classical D-optimality and $\mathcal{I}_0(\pmb{\xi}) $ the classical I-optimality. 
For given $\nu \in (0, 1]$ a design $\pmb{\xi} $ on $\chi$ is defined to be \textbf{I-robust} if it minimizes $\mathcal{I}_\nu(\pmb{\xi})$ in the class of all designs on $\chi$, and \textbf{D-robust} if it minimizes $\mathcal{D}_\nu(\pmb{\xi})$.
If $\nu=1$ then the uniform design is D-/I- robust. 

The algorithm proposed in~\cite{wiens} is based on the QR-decomposition of $\pmb{F}$, 
where $\pmb{Q}$ is the $Q$-matrix in such decomposition, 
and finally 
\[
\pmb{R}(\pmb{\xi}) = \pmb{Q}' \pmb{D}(\pmb{\xi}) \pmb{Q} \quad \text{ and } \quad
\pmb{U}(\pmb{\xi}) = \pmb{R}^{-1}(\pmb{\xi})   \pmb{Q}' \pmb{D}^2(\pmb{\xi}) \pmb{Q}    \pmb{R}^{-1}(\pmb{\xi})
\]

\begin{algorithm}
\caption{Wiens approach for D-robustness and $\nu\neq 0,1$}\label{Wiens}
\begin{algorithmic}[1]
\State Let the sample space be $\mathcal X=\{\pmb{ x_i} \}_{i=1,\ldots,N}$ and $\pmb{e}_i$ the $i$-th column of the $I_{N}$ identity matrix, QR-decomposition of $\pmb{F}$
\State Sample randomly a subset of size $n_t < n_d$ from the sample space
($n_t=1$ is ok). 
Set the current sample size $n= n_t$ and $\pmb{\xi}=\left\{ \begin{matrix}
\pmb{x} \\ \pmb{\xi_{n,\pmb{x}}}
\end{matrix} 
\right\}_{\pmb{x} \in \mathcal{X}}$
\While  {$n \le n_d$ or when a certain criteria is not met} 
\State 
compute{  \small{
\begin{itemize} 
\item $\lambda$ maximum eigenvalue of $
  \pmb{R}^{1/2}(\pmb{\xi})[\pmb{U}(\pmb{\xi})-\pmb{I}_p]R^{1/2}(\pmb{\xi})
$

\item and  $\pmb{z}$  the corresponding eigenvector  

\item the vectors $
\pmb{v}(\pmb{\xi}) {=} \pmb{R}^{1/2}(\pmb{\xi})\pmb{z}(\pmb{\xi})$ and $
\pmb{w}(\pmb{\xi}) {=} \pmb{R}^{-1/2}(\pmb{\xi})\pmb{z}(\pmb{\xi})$
\item the matrices 
\begin{align*}
\pmb{J}(\pmb{\xi}) &= \lambda \left(\pmb{R}^{-1}(\pmb{\xi}) + \pmb{w}(\pmb{\xi})\pmb{w}'(\pmb{\xi}) \right) + \left( \pmb{w}(\pmb{\xi})\pmb{v}'(\pmb{\xi}) + \pmb{v}(\pmb{\xi})\pmb{w}'(\pmb{\xi}) \right)\\
\pmb{K}(\pmb{\xi}) &= 2 \pmb{w}(\pmb{\xi})\pmb{w}'(\pmb{\xi})
\end{align*}
\item and 
\[ T(\pmb{\xi}) = (1-\nu)\pmb{Q} \pmb{R}^{-1}(\pmb{\xi}) \pmb{Q}' + \nu [\pmb{Q} \pmb{J}(\pmb{\xi}) \pmb{Q}' - \pmb{D}(\pmb{\xi})\pmb{Q} \pmb{K}(\pmb{\xi}) \pmb{Q}'],
\]
\item the largest diagonal element of $T(\pmb{\xi}) $ and assume it is in entry $(i,i)$ 
\end{itemize}   
}}

\State Update the weights of  $\pmb{\xi_{n,\pmb{x}}}$ to
$
 \pmb{\xi_{n+1,\pmb{x}}} = \left(\frac{n}{n+1}\right) \left(\pmb{\xi}_n + \frac{1}{n}\pmb{e}_i\right) $

\State \textbf{go to while}.
\EndWhile
\end{algorithmic}
\end{algorithm}

The pseudocode is given in Algorithm~\ref{Wiens}.
Its main features are that bias is accounted for and the optimal design is known prior observing. 
The method is supported by strong theoretical background. 
As given the algorithm is purely sequential, but it can be tweaked to become adaptive. 
Unfortunately it requires the QR-decomposition of a high dimensional matrix and requires to keep in memory large matrices. 
Current available implementation is not very performing but a smart implementation may overcome some of these issues and make the algorithm efficient for significatively large sample sizes.

%%%%%%%%%
%%%%%%%%%
%%%%%%%%%
%%%%%%%%%
%%%%%%%%%

Next we recall some precursory work on optimal subdata collection for linear regression  based on the information matrix~\cite{montepiedra}
 which, we believe, is useful in the presence of big data.
 For $\xi_{ \pmb{x}}$ as above, %a design measure on a (discrete) $\mathcal{X}$, e.g. $\mathcal{X} = \textbf{Data}$, 
the information matrix can be written as
\[ M =  \int_{\mathcal X}    \left( \begin{matrix} \pmb{f} \\ \pmb{h} \end{matrix} \right) \left( \pmb{f}^\prime, \pmb{h}^\prime \right)
\, d\,\xi_{\pmb{x}} 
 =
\left [ \begin{array}{cc}
M_{11} & M_{12} \\
M_{21} & M_{22} \end{array} \right] 
\]
where $M_{11}$ depends only on $\pmb{f} $ and $M_{22}$ depends only on $\pmb{g} $.
The mean square error of the LSE of $\theta$ is $\mathbb{E} \{ (\hat{\pmb{\theta}} - \pmb{\theta})(\hat{\pmb{\theta}} - \pmb{\theta})^\prime \} {=} \sigma^2 N^{-1} R$
where $$ R = M_{11}^{-1} {+} \left( \frac{N}{\sigma} \right)^2  M_{11}^{-1} M_{12} \pmb{\psi} \pmb{\psi}^\prime M_{21}M_{11}^{-1} $$
 and the loss function for $D$-optimality becomes 
 \[ \operatorname{det} (R) = \operatorname{det} (M_{11}^{-1}) \,
\left (1+ \left( \frac{N}{\sigma} \right)^2   \pmb{\psi}^\prime M_{21} M_{11}^{-1} M_{12} \pmb{\psi} \right) \]
A $D$-optimal design for bias reduction satisfies the following optimisation problems
%\begin{description}
%\item[(bias reduction):] $ ~ ~
\[
 \xi^\ast = \underset{\xi}{\operatorname{argmax}} \operatorname{det}(M_{11}) \text{ such that }\left( \frac{N}{\sigma} \right)^2 \pmb{\psi}^\prime M_{21} M_{11}^{-1} M_{12} \pmb{\psi} \le B
 \]
 for a given $B$. 
(The authors in~\cite{montepiedra} also study variance reduction but here we focus on bias reduction as far more relevant in the analysis of big data.) 
%\item[(variance reduction):] $ ~ ~ \xi^\star = \underset{\xi}{\operatorname{argmin}} \left( \frac{N}{\sigma} \right)^2 \pmb{\psi}^\prime M_{21} M_{11}^{-1} M_{12} \pmb{\psi} $ such that $ \operatorname{det}(M_{11}) \ge D$
Thus the objective becomes to determine 
$\xi^\ast = \operatorname{argmax}_{\xi} \operatorname{det}(M_{11})$ such that $\left( \frac{N}{\sigma} \right)^2   \pmb{\psi}^\prime M_{21} M_{11}^{-1} M_{12} \pmb{\psi} \le B$.
% In~\cite{montepiedra} it is shown that a
A design $\xi^\ast$ is optimal 
 if and only if there exists $\lambda^\ast \ge 0$ such that
\[
d_1(\pmb{x}, \xi^\ast) + \lambda^\ast d_2(\pmb{x}, \xi^\ast) \le  p - \lambda^\ast B
\]
for all $\pmb{x} \in \mathcal{X}$, where
\begin{align*}
d_1(\pmb{x}, \xi) &= \pmb{f}(\pmb{x})' M_{11}^{-1} \pmb{f}(\pmb{x}) \, \, \text{   and    } \, \, d_2(\pmb{x}, \xi) = \phi^2(\pmb{x}, \xi) - 2 \phi(\pmb{x}, \xi) r(\pmb{x})
\end{align*}
with
\begin{align*}
r(\pmb{x}) &= \frac{N}{\sigma} \pmb{\psi}'  \pmb{h}(\pmb{x}) \, \, \text{ and } \, \, \phi(\pmb{x},\xi) = \int_{\mathcal{X}} \pmb{f}(\pmb{x})' M_{11}^{-1} \pmb{f}(\pmb{x}') \xi(d\pmb{x}')
\end{align*}

This gives a strong theoretical background and makes a good link with the next sections but does not provide specific algorithms nor applies to big data directly.

\subsection{Models with no bias terms}
\label{section_onlyconfounders}

Models of the form 
$\pmb{Y}(\pmb{x, z}) = \pmb{f}(\pmb{x})^\prime \pmb{\theta} + \pmb{g}(\pmb{z})^\prime \pmb{\phi} + \pmb{\epsilon}$
have been considered in~\cite{pescecladag}.
%The two main points of~\cite{pescecladag}  
%consider graphical models and bias models; the former is relevant in order to study conditional independence structures among problem variables, the latter is needed to ascertain and guard against hidden sources of bias. Using this approach, it is possible to demonstrate the existence of a Nash equilibrium in an optimal design setting, and this can be extended to the case of randomization, since it can be considered as a mixed strategy in game theoretic approach.  % Here we do not develop this further.
%
 Let $\xi_{ \pmb{x,z}}$ be a design measure on $\mathcal{X} \times \mathcal{Z}$. The information matrix is 
 \[
 M =  \mathlarger{\int}_{\mathcal{X} \times \mathcal{Z}} \left( \begin{matrix} \pmb{f} \\ \pmb{g} \end{matrix} \right) \left( \pmb{f}^\prime, \pmb{g}^\prime \right)
\, d\,\xi_{\pmb{x,z}} \]
The mean square error of the LSE of $\theta$ is $ \sigma^2 N^{-1}  R$
where 
\[ R = M_{11}^{-1} {+} \left( \frac{N}{\sigma} \right)^2  M_{11}^{-1} M_{12} \pmb{\phi} \pmb{\phi}^\prime M_{21}M_{11}^{-1}
\]
and the loss function for $D$-optimality depends on both $\pmb{x}$ and $\pmb{z}$ and it is
\[ \operatorname{det} (R) = \operatorname{det} (M_{11}^{-1}) \,
\left (1+ \left( \frac{N}{\sigma} \right)^2   \pmb{\phi}^\prime M_{21} M_{11}^{-1} M_{12} \pmb{\phi} \right) 
\] 

 We assume $\pmb{g}(\pmb{z})^\prime \pmb{\phi}$ unknown, belonging to some function class. For each $\pmb{x} \in \textbf{Data}$ there is an unobserved $\pmb{z_x} \in \mathcal{Z}$. Let $\pmb{G} =  [\pmb{g}(\pmb{z_x})]_{\pmb{x} \in \textbf{Data}}$ and $P_Z$ be a randomization distribution for the $\pmb{z_x}$'s. 
In the game theoretical approach in~\cite{pescecladag}, a $D$-optimal design measure is one maximising
\[
\underset{P_Z}{\min} ~ \mathbb{E}_{P_Z} \left\{ \max_{\substack{\text{ \small{function}}\\ \text{\small{class}}}} \, \pmb{G}^\prime \pmb{\phi} \pmb{F} M_{11}^{-2} \pmb{F}^\prime \pmb{\phi}^\prime \pmb{G} \right\}
\]
% So far this is only theoretical.

%%%%%%%%%
%%%%%%%%%
%%%%%%%%%
%%%%%%%%%
%%%%%%%%%

\section{General formulation: model with bias and confounders}
\label{general_formulation}

In this section we consider the more general form for the response variable $\pmb{Y}$ in Model~(\ref{model}), 
 define  the variance function to be $d((\pmb{x},\pmb{z}), \pmb{\xi}) = \pmb{f}^\prime(\pmb{x}) M(\pmb{\xi})^{-1} \pmb{f}(\pmb{x})$ %; this is called the variance function. 
and assume that $M(\pmb{\xi})$ is a closed and bounded subset of the semi-definite positive matrices. 
Then a version of the General Equivalence Theorem for Model~(\ref{model}) holds.

\begin{theorem}
\label{get_theorem}
For a design measure $\pmb{\xi}^\ast$ the following statements are equivalent 
\begin{enumerate}[label=(\roman*)]
%\emph{(i)}
\item $\pmb{\xi}^\ast$ maximizes $\operatorname{det}(M(\pmb{\xi}))$
%\emph{(ii)} 
\item $\pmb{\xi}^\ast$ achieves $\min_{\pmb{\xi}} \max_{(x,z) \in \mathcal{X} \times \mathcal{Z}} d((\pmb{x},\pmb{z}), \pmb{\xi})$
% \emph{(iii)} 
\item $\max_{(x,z) \in \mathcal{X} \times \mathcal{Z}} d((\pmb{x},\pmb{z}), \pmb{\xi}^\ast) = p+m+q-2$.
\end{enumerate}
\end{theorem}

%The idea behind the proof of Theorem~\ref{get_theorem} derives from the fact that the Model~(\ref{model}) can be seen as the sum of three linear models (the three parts that compose it) and supposing that $\pmb{\xi}_1^\ast$, $\pmb{\xi}_2^\ast$ and $\pmb{\xi}_3^\ast$ are D-optimal, respectively, for $\pmb{f}^\prime(\pmb{x}) \pmb{\theta}$, $\pmb{h}^\prime(\pmb{x}) \pmb{\psi}$ and $\pmb{g}^\prime (\pmb{z}) \pmb{\phi}$, then the product design $\pmb{\xi}_1^\ast \otimes \pmb{\xi}_2^\ast \otimes \pmb{\xi}_3^\ast$ is D-optimal for the model sum in Equation~(\ref{model}).

\begin{proof} 
There are two further equivalent conditions which allows a circular proof.
One is a local D-optimality condition
\[
\emph{(iv)} 
\frac{\partial}{\partial(\pmb{x},\pmb{z})} \log \operatorname{det} \left(  M \left( \left( 1- \alpha \right) \pmb{\xi}^\ast + \alpha  \pmb{\xi}^\prime  \right) \right)\mid_{\alpha = 0} \, \le \, \pmb{0}
\]
where $\pmb{\xi}^\ast $ and $ \pmb{\xi}^\prime $ are design measures. 
The fifth equivalent condition is  
\[
\emph{(v)} \, \, d((\pmb{x},\pmb{z}), \pmb{\xi}) \le p+m+q-2 \, \, \, \text{ for all } (\pmb{x},\pmb{z}) \in \mathcal{X} \times \mathcal{Z}
\]

That Item~\emph{(i)} implies \emph{(iv)} is straigthforward.
To show that Item~\emph{(iv)} implies \emph{(v)} we use the matrix identity 
$
\frac{\partial}{\partial \alpha} \log \operatorname{det}(A) = \operatorname{tr} \left( A^{-1} \frac{\partial A}{\partial \alpha} \right).
$. 
Thus
\begin{align*}
\frac{\partial}{\partial(\pmb{x},\pmb{z})} \log & \operatorname{det} \left(  M \left( \left( 1- \alpha \right) \pmb{\xi}^\ast + \alpha  \pmb{\xi}^\prime  \right) \right)\mid_{\alpha = 0}\\
&= \operatorname{tr} \left( M^{-1} \left( \left( 1- \alpha \right) \pmb{\xi}^\ast + \alpha  \pmb{\xi}^\prime \right) \cdot \frac{\partial}{\partial \alpha} M \left( \left( 1- \alpha \right) \pmb{\xi}^\ast + \alpha  \pmb{\xi}^\prime \right) \right)\mid_{\alpha = 0}\\
&= \operatorname{tr} \left( M^{-1} \left( \left( 1- \alpha \right) \pmb{\xi}^\ast + \alpha  \pmb{\xi}^\prime \right) \cdot \frac{\partial}{\partial \alpha}  \left( (1 - \alpha) M(\pmb{\xi}^\ast) + \alpha  M(\pmb{\xi}^\prime) \right) \right)\mid_{\alpha = 0}\\
&= \operatorname{tr} \left( M^{-1} \left( \left( 1- \alpha \right) \pmb{\xi}^\ast + \alpha  \pmb{\xi}^\prime \right) \cdot   \left( - M(\pmb{\xi}^\ast) +  M(\pmb{\xi}^\prime) \right) \right)\mid_{\alpha = 0}\\
&= \operatorname{tr}\left( M^{-1}(\pmb{\xi}^\ast) \cdot (-M(\pmb{\xi}^\ast) + M(\pmb{\xi}^\prime)) \right)\
= \operatorname{tr}\left( - I_{p+m+q-2} +   M^{-1}(\pmb{\xi}^\ast) M(\pmb{\xi}^\prime) \right)\\
&= -(p+m+q-2) + \operatorname{tr} \left(  \int_{\mathcal{X} \times \mathcal{Z}} d((\pmb{x},\pmb{z}), \pmb{\xi}^\ast) \, \pmb{\xi}^\prime(dx , dz) \right)
\end{align*}
so that the statement in \emph{(iv)} is equivalent to $ \int_{\mathcal{X} \times \mathcal{Z}} d((\pmb{x},\pmb{z}), \pmb{\xi}^\ast) \, \pmb{\xi}^\prime(dx , dz) \le p+m+q-2$ for all $ \pmb{\xi}^\prime$. This holds in particular when $\pmb{\xi}^\prime$ places mass one at a specific point $(\pmb{x},\pmb{z})$. 
But this is $d((\pmb{x},\pmb{z}), \pmb{\xi}^\ast) \le p+m+q-2 $ for all $ (\pmb{x},\pmb{z})$, so \emph{(v)} is verified.

To prove that 
 \emph{(iii)} is equivalent to \emph{(iv)}, we can show that 
 \[ \max_{(x,z) \in \mathcal{X} \times \mathcal{Z}} d((\pmb{x},\pmb{z}), \pmb{\xi}^\ast) \ge p+m+q-2 \text{ for all } (\pmb{x},\pmb{z}) 
 \]
  But this follows from  the fact that a maximum is always greater than or equal to an average, so
\begin{align*}
\max_{(x,z) \in \mathcal{X} \times \mathcal{Z}} d((\pmb{x},\pmb{z}), \pmb{\xi}^\ast) &\ge \int_{\mathcal{X} \times \mathcal{Z}} d((\pmb{x},\pmb{z}), \pmb{\xi}^\ast) \, \pmb{\xi}^\ast(dx , dz)\\
&= \operatorname{tr} \left( M^{-1}(\pmb{\xi}^\ast) M(\pmb{\xi}^\ast) \right) = \operatorname{tr}(I_{p+m+q-2}) = p+m+q-2 
\end{align*}
%To be a little more precise we need to exhibit a $\pmb{\xi}^\ast$ actually achieving the bound, and this is possible  because we assumed that $M(\pmb{\xi})$ is a closed and bounded subset of the semi-definite positive matrices.
As we assumed that $M(\pmb{\xi})$ is a closed and bounded subset of the semi-definite positive matrices, $\pmb{\xi}^\ast$  achieves the bound

% The remaining part is to prove \emph{(iii)} $\iff$ \emph{(i)}. Suppose \emph{(iii)} holds. 
Lastly we prove that \emph{(iii)}  implies \emph{(i)}.
We use the identity $\operatorname{tr}(A) \ge n \cdot \operatorname{det}(A)^{\frac{1}{n}}$ for an $n \times n$ matrix $A$. 
Thus for $k = p+m+q-2$ we have
\begin{align*}
k &= \max_{(x,z) \in \mathcal{X} \times \mathcal{Z}} d((\pmb{x},\pmb{z}), \pmb{\xi}^\ast) \ge \int_{\mathcal{X} \times \mathcal{Z}} d((\pmb{x},\pmb{z}), \pmb{\xi}^\ast) \, \pmb{\xi}^\prime(dx , dz)\\
&= \operatorname{tr} \left( M^{-1}(\pmb{\xi}^\ast) M(\pmb{\xi}^\prime) \right) \ge k \cdot \operatorname{det} \left( M^{-1}(\pmb{\xi}^\ast) M(\pmb{\xi}^\prime) \right)^{\frac{1}{k}}\\
&= k \cdot  \left( \operatorname{det} \left( M^{-1}(\pmb{\xi}^\ast) \right) \operatorname{det} \left( M(\pmb{\xi}^\prime) \right) \right)^{\frac{1}{k}} = k \cdot \left( \frac{\operatorname{det}(M(\pmb{\xi}^\prime))}{\operatorname{det}(M(\pmb{\xi}^\ast))}  \right)^{\frac{1}{k}}
\end{align*}
From this $\operatorname{det}(M(\pmb{\xi}^\ast)) \ge \operatorname{det}(M(\pmb{\xi}^\prime))$, which is \emph{(i)}; so we have \emph{(iii)} holds if and only if \emph{(i)} holds. \qed
\end{proof}

\subsection{Guard against bias}
\label{bias}

In analogy to Subsection~\ref{section_onlyconfounders},
we want to protect the usual LSE of $\pmb{\theta}$ in Model~(\ref{model}) against the two bias terms $\pmb{\psi}$ and $\pmb{\phi}$. 
The information matrix for a $\xi_{\pmb{x},\pmb{z}}$ design measure on $\mathcal{X} \times \mathcal{Z}$ can be written as 
\begin{equation*}
%\label{matrix}
M = M(\pmb{\xi}) = \mathlarger{\int}_{\mathcal{X} \times \mathcal{Z}} \left( \begin{matrix} \pmb{f} \\ \pmb{h} \\ \pmb{g} \end{matrix} \right) \left( \pmb{f}^\prime, \pmb{h}^\prime, \pmb{g}^\prime \right)\, d\,\xi_{\pmb{x,z}} =
\left [ \begin{array}{ccc}
M_{11} & M_{12} & M_{13}\\
M_{21} & M_{22} & M_{23}\\
M_{31} & M_{32} & M_{33} \end{array} \right] 
\end{equation*}
where by symmetry $M_{12}$ is the transpose of $M_{21}$ and so on. 
The mean square error matrix of the LSE of $\pmb{\theta}$ is $ \sigma^2 N^{-1} R$ where $N$ is the sample size, $\sigma^2$ the common variance of the error terms for the Model~(\ref{model}) and where 
\[  
R = M_{11}^{-1} {+} \left( \frac{N}{\sigma} \right)^2  M_{11}^{-1} \left[ M_{12} ~ M_{13}\right] \left[\pmb{\psi} ~ \pmb{\phi} \right]^\prime \left[\pmb{\psi}^\prime ~ \pmb{\phi}^\prime \right] 
\left[ \begin{array}{c} M_{21} \\ M_{31} \end{array} \right]^\prime 
M_{11}^{-1}
\]

Above we gave an elementary proof of a General Equivalence Theorem in order to get a relation between optimality criteria (D-, G- and A-optimality).
% that ask to minimize over the choice of experimental design a loss function of the information matrix $M(\pmb{\xi})$. 
In this subsection we are interested in minimising loss functions of the matrix $R$. Future work will focus on making a relation between the loss functions of $M(\pmb{\xi})$ and $R$, in order to use the General Equivalence Theorem also for the matrix $R$. 
In particular, here we concentrate on the A-optimality and derive a formula for $\operatorname{tr}(R)$ 

\begin{align*} 
\operatorname{tr}(R) =& \operatorname{tr}\left( M_{11}^{-1} {+} \left( \frac{N}{\sigma} \right)^2  M_{11}^{-1} \left[ M_{12} ~ M_{13}\right] \left[\pmb{\psi} ~ \pmb{\phi} \right]^\prime \left[\pmb{\psi}^\prime ~ \pmb{\phi}^\prime \right] \left[ M_{21} ~ M_{31}\right]^\prime M_{11}^{-1} \right)\\
=& \operatorname{tr}\left(M_{11}^{-1}\right) + \left( \frac{N}{\sigma} \right)^2  \operatorname{tr} \left(  M_{11}^{-1} \right. \\
& \left. \left[ M_{12} \pmb{\psi} \pmb{\psi}^\prime M_{21} + M_{13} \pmb{\phi} \pmb{\phi}^\prime M_{31} + M_{12} \pmb{\psi} \pmb{\phi}^\prime M_{31} + \left(M_{12} \pmb{\psi} \pmb{\phi}^\prime M_{31} \right)^\prime \right] M_{11}^{-1} \right)\\ 
%=& \operatorname{tr}\left(M_{11}^{-1}\right) + \left( \frac{N}{\sigma} \right)^2 \operatorname{tr}\left(M_{11}^{-1} M_{12} \pmb{\psi} \pmb{\psi}^\prime M_{21} M_{11}^{-1} \right) + \left( \frac{N}{\sigma} \right)^2 \operatorname{tr}\left(M_{11}^{-1} M_{13} \pmb{\phi} \pmb{\phi}^\prime M_{31} M_{11}^{-1} \right)\\
%&+ \left( \frac{N}{\sigma} \right)^2 \operatorname{tr}\left(M_{11}^{-1} M_{12} \pmb{\psi} \pmb{\phi}^\prime M_{31} M_{11}^{-1} \right) + \left( \frac{N}{\sigma} \right)^2 \operatorname{tr}\left(M_{11}^{-1} \left( M_{12} \pmb{\psi} \pmb{\phi}^\prime M_{31} \right)^\prime M_{11}^{-1} \right)\\
%=& \operatorname{tr}\left(M_{11}^{-1}\right) + \left( \frac{N}{\sigma} \right)^2 \operatorname{tr}\left(M_{11}^{-1} M_{12} \pmb{\psi} \pmb{\psi}^\prime M_{21} M_{11}^{-1} \right) + \left( \frac{N}{\sigma} \right)^2 \operatorname{tr}\left(M_{11}^{-1} M_{13} \pmb{\phi} \pmb{\phi}^\prime M_{31} M_{11}^{-1} \right)\\
%&+ \left( \frac{N}{\sigma} \right)^2 \operatorname{tr}\left(M_{11}^{-1} M_{12} \pmb{\psi} \pmb{\phi}^\prime M_{31} M_{11}^{-1} \right) + \left( \frac{N}{\sigma} \right)^2 \operatorname{tr}\left(M_{11}^{-1} M_{12} \pmb{\psi} \pmb{\phi}^\prime M_{31} M_{11}^{-1} \right)\\
=& \operatorname{tr}\left(M_{11}^{-1}\right) + \left( \frac{N}{\sigma} \right)^2 \operatorname{tr}\left(M_{11}^{-1} M_{12} \pmb{\psi} \pmb{\psi}^\prime M_{21} M_{11}^{-1} \right) + \left( \frac{N}{\sigma} \right)^2 \operatorname{tr}\left(M_{11}^{-1} M_{13} \pmb{\phi} \pmb{\phi}^\prime M_{31} M_{11}^{-1} \right)\\
&+ 2 \left( \frac{N}{\sigma} \right)^2 \operatorname{tr}\left(M_{11}^{-1} M_{12} \pmb{\psi} \pmb{\phi}^\prime M_{31} M_{11}^{-1} \right)\\
=& \operatorname{tr}(M_{11}^{-1}) + \left( \frac{N}{\sigma} \right)^2 \operatorname{tr}(S_2) + \left( \frac{N}{\sigma} \right)^2 \operatorname{tr}(S_3) + 2 \left( \frac{N}{\sigma} \right)^2 \operatorname{tr}(S_4)
\end{align*}
where
\begin{align*}
\operatorname{tr}(S_2) &= \operatorname{tr}\left(M_{11}^{-1} M_{12} \pmb{\psi} \pmb{\psi}^\prime M_{21} M_{11}^{-1} \right) = \operatorname{tr}\left( \pmb{\psi}^\prime M_{21} M_{11}^{-1} M_{11}^{-1} M_{12} \pmb{\psi}  \right)\\ &= \pmb{\psi}^\prime M_{21} M_{11}^{-2} M_{12} \pmb{\psi}\\
\operatorname{tr}(S_3) &= \pmb{\phi}^\prime M_{31} M_{11}^{-2} M_{13} \pmb{\phi}\\
\operatorname{tr}(S_4) &= \pmb{\phi}^\prime M_{31} M_{11}^{-2} M_{12} \pmb{\psi}
\end{align*}
and thus
\begin{align}
\label{traceR}
\operatorname{tr}(R) &=\operatorname{tr} \left( M_{11}^{-1}\right)  \\
& + \left( \frac{N}{\sigma} \right)^2 \left( \pmb{\psi}^\prime M_{21} M_{11}^{-2} M_{12} \pmb{\psi} 
+  \pmb{\phi}^\prime M_{31} M_{11}^{-2} M_{13} \pmb{\phi} + \pmb{\phi}^\prime M_{31} M_{11}^{-2} M_{12} \pmb{\psi} \right)  \nonumber
\end{align}

In $\operatorname{tr}(R)$, the first and second terms depend  only on $\pmb{x}$ and te third term on $\pmb{x}$ and on $\pmb{z}$ but not on the bias $\pmb{h}$.
When the foruth term is equal to zero, then the minimization of the $\operatorname{tr}(R)$ can be done
separately on the $\pmb{x}$ variables and the $\pmb{z}$ variables and generalization to non linear confounders is easier. 
%This would also facilitate the analysis of models where the confounder term $ \pmb{g}(\pmb{z})^\prime \pmb{\phi} $ is not the linear function of unknown parameters but a generic function of the $\pmb{z}$ variables.

%%%%%%%%%
%%%%%%%%%
%%%%%%%%%
%%%%%%%%%
%%%%%%%%%

\section{Examples and simulations}

\begin{example} \rm{ 
%\subsection{Case study: mortgage dataset}
\label{mortgage}

% Algorithms~\ref{Drovandi} and~\ref{Stufken} are compared on the simulated mortgage defaults (year 2000) dataset analysed in~\cite{drovandi}.
Algorithm~\ref{Drovandi} is applied on the simulated mortgage defaults (year 2000) dataset analysed in~\cite{drovandi}.
The dataset has $1,000,000$ data points, a binary response for the mortgage default $Y_i \sim \text{\emph{Binary}}(\pi_i)$ and four covariates: 
credit Score ($\pmb{x}_1$), 
age of the house in years ($\pmb{x}_2$),
number of years the mortgage holder has been employed at current job ($\pmb{x}_3$)
and amount of credit card debt ($\pmb{x}_4$).
The scaled values of the data points are clustered around the grid in Table~\ref{grid}, so we take this as the grid used in Algorithm~\ref{Drovandi}. 
The response is skewed: $Y_i=1$ in 1031 units and $Y_i=0$ for  $998,969$ units and following~\cite{drovandi} 
we assume  $Y_i \sim \text{\emph{Binary}}(\pi_i)$ and a logistic model  $\text{\emph{logit}}(\pi_i) = \theta_0 + \theta_1 x_{1i} + \theta_2 x_{2i} + \theta_3 x_{3i} + \theta_4 x_{4i}$.

\begin{table}[h]
%\centering
\begin{tabular}{l|c}
\hline
\textbf{Covariate} & \textbf{Grid}\\
\hline
creditscore & -4, -3, -2, -1, 0, 1, 2, 3, 4\\
houseAge & -2, -1, 0, 1, 2\\
yearsemploy & -2, -1, 0, 1, 2, 3, 4\\
ccDebt & -2, -1, 0, 1, 2, 3, 4\\
\hline
\end{tabular}
\caption{Grid generated by data points.}
\label{grid}
\end{table}

The maximum likelihood estimates of the parameters of the logit models are obtained % using Algorithm~\ref{Drovandi}
starting with $n_t =5,000$ points in step 2. and with a final sample of size $n_d = 6,200$. 
The estimates are consistent with those in~\cite{drovandi}. 

Further to~\cite{drovandi}
in Figure~\ref{train} we investigate the effect of the choice of the initial sample on the parameter estimates:
 \textbf{black} refers to an initial sample including all units for which $Y=1$ (a dope training set in machine learning), 
\textbf{red}  to a randomly selected initial sample,
\textbf{green} to a stratified sample: we considered the distribution of ccDebt for the sub-population for which $Y=1$ and sampled one data points for each quantile, since preliminary analysis indicates that ccDebt  effect most the response. 
All the estimates converge to the same values, but the black being quicker as expected. % the black line is the one that converges quicker. Hence if we choose $n_d$ appropriately, the choice of the initial sample is not relevant. 

%We use Algorithm~\ref{Drovandi} in order to obtain a MLE of the parameters: we choose to start with $n_t = 5,000$ initial sample and we want a final sample of $n_d = 6,200$ data points. First we want to investigate if different choices for the selection of the initial sample affect the final estimates, especially since we have a skewed response. The plots in Figure~\ref{train} refers to the estimated values of parameters at each iterations of the algorithm. The colors represent different choices of the initial sample: \textbf{black} refers to a sample with all the ``1'' contained in it (what in machine learning we may call a dope training set), \textcolor{red}{\textbf{red}} corresponds to a completely random selection of the initial sample, while for the \textcolor{green}{\textbf{green}} lines we considered the distribution of ccDebt for the sub-population of the ``1'' and sample one data points for each quantile, since this covariate seems to be the one that effects more the response after some preliminary descriptive analysis. We can observe that all the estimates converges to the same values, but as we expected the black line is the one that converges quicker. Hence if we choose $n_d$ appropriately, the choice of the initial sample is not relevant. 

\begin{figure}[h]
\includegraphics[width=\textwidth]{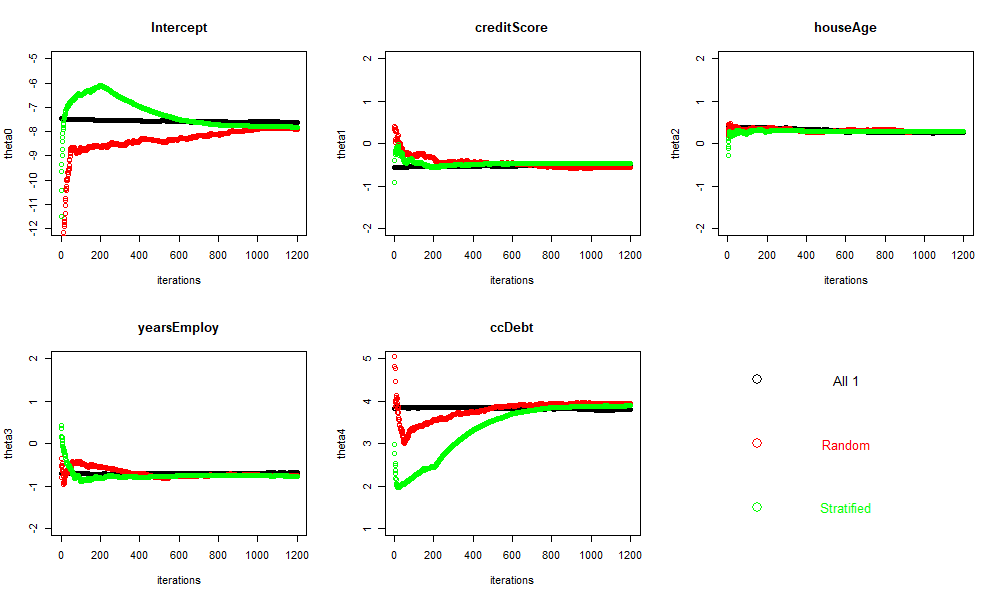}
\caption{Effect of the initial sample on the output of Algorithm~\ref{Drovandi}}
\label{train}
\end{figure} 

The performance of Algorithm~\ref{Drovandi}  in terms of the prediction of the response outcomes is tested on $10,010$ data points that are not considered above. 
The comparison is made with random forests (RF) and neural networks (NN) build with a random training set or with the final ``best'' sample obtained through Algorithm~\ref{Drovandi}. The results are report in Table~\ref{confusion_matrix}.  
Algorithm~\ref{Drovandi} performs better than the other approaches with a random training set, but the performance of  a RF or a NN is better when starting with  
  the best sample from Algorithm~\ref{Drovandi}. 

%We are also interested on how Algorithm~\ref{Drovandi} performs in terms of the prediction of the response outcomes. We consider $10,010$ data points that are not considered in the analysis above as a test set. We compare the predictions of Algorithm~\ref{Drovandi} with respect to random forests (RF) build with a random training set or with the final ``best'' sample obtained through Algorithm~\ref{Drovandi}, the same for neural networks (NN). In Table~\ref{confusion_matrix} are reported the confusion matrices of each approach: Algorithm~\ref{Drovandi} performs very well with respect to the other approaches with a random train, but when we use the best sample we improve the performance of the random forest and neural network.

\begin{table}[h]
%\centering
\begin{tabular}{l|c}
\hline
\textbf{Model} & \textbf{Confusion matrix}\\
\hline
Algorithm 1 & $\begin{bmatrix}
 9765  &  3\\
 235  &  7\\
\end{bmatrix}$ \\
\hline
RF + random train & $\begin{bmatrix}
10000  &  10\\
0  &  0\\
\end{bmatrix}$\\
\hline
RF + best sample &  $\begin{bmatrix}
9718  &  2\\
282  &  8\\
\end{bmatrix}$\\
\hline
NN + random train &  $\begin{bmatrix}
10000  &  10\\
0  &  0\\
\end{bmatrix}$\\
   \hline
NN + best sample &  $\begin{bmatrix}
9176  &  0\\
824  &  10\\
\end{bmatrix}$\\
\hline
\end{tabular}
\caption{Confusion matrices of different methods.}
\label{confusion_matrix}
\end{table}

} \end{example} 

%%%%%%%%%
%%%%%%%%%
%%%%%%%%%
%%%%%%%%%
%%%%%%%%%

\begin{example} \label{DrSt}
\rm{ 

Algorithms~\ref{Drovandi} and~\ref{Stufken} are compared on simulated data  from the model
 $y=-x/2-5/3+0.35\, \sin( x^2) + z/9 + \epsilon $
 which includes both the bias term $ 0.35\, \sin( x^2) $ and the confounder term $g(z)=z/9$.  
Hundred and five points in $\mathbb R^2$ were generated from two independent gaussian random variables $X\sim \mathcal N(2,1)$ and $Z\sim \mathcal N(0,1)$ for the first and second component of the points, respectively.
% for $z$ hundred points were generated from an independent standard normal. 
The grid used in the algorithms is given by  $200$ uniformly distributed points for $x$, chosen between the minimum and the maximum  generated
values, and crossed with $200$ points for $z$ chosen in the same manner. 

Figure~\ref{aa} shows the twelve point optimal designs returned by the two algorithms (in green Algorithm~\ref{Drovandi} and in red Algorithm~\ref{Stufken})
when the $D$-optimality utility function is computed on $-x/2-5/3+ z/9 $. 
Algorithm~\ref{Stufken} pushes the selected points more on the boundary of the $(x,z)$-plane. 
The plot in the left panel of  Figure~\ref{bb}  projects the designs in Figure~\ref{aa} of the response-$x$ plane, 
the right plot compares on the same plane the ``optimal'' designs returned by the two algorithms when all biases are ignored and  the optimality function is thus computed on $-x/2-5/3$. 
As expected the outputs for the models with no confounders are very similar, begin different in just one point. 
Always the value of the utility function is larger for Algorithm~\ref{Stufken}.

\begin{figure}[h]  
\includegraphics[width=0.75\textwidth]{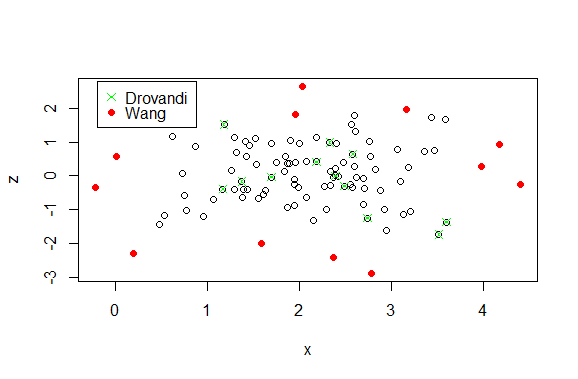}    % aa = drovandiWang_x_vs_z.png
\caption{Outputs from Algorithms~\ref{Drovandi} in red and~\ref{Stufken} in green for Example~\ref{DrSt}}
\label{aa}
\end{figure}

\begin{figure}[h]  
\includegraphics[width=0.5\textwidth]{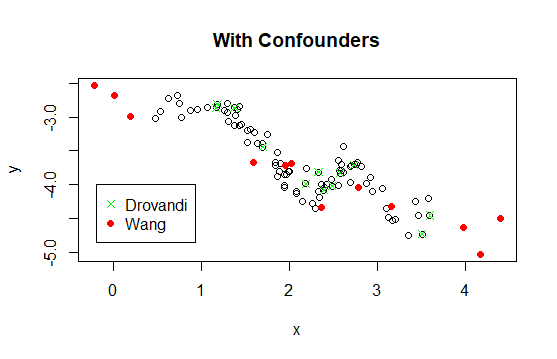}  % bb = drovandiStufken_bias_noBias    
\includegraphics[width=0.5\textwidth]{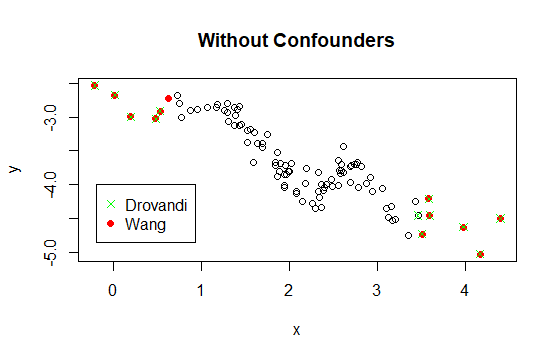}

\caption{Different utilities  for Example~\ref{DrSt}}
\label{bb}
\end{figure} 

} \end{example} 

%%%%%%%%%
%%%%%%%%%
%%%%%%%%%
%%%%%%%%%
%%%%%%%%% 

\begin{example} \label{DrStWiens}  \rm{ 
Next we consider $105$ integer in $-100{:}100$ and the model $x+\cos(x) + z/9+\epsilon$ with $\epsilon \sim \mathcal N(0,1)$.
The $D$-optimal designs returned by Algorithms~\ref{Drovandi},~\ref{Stufken} and~\ref{Wiens} are plotted on the $x{-}z$ plane in Figure~\ref{cc}.
In the right panel the utility function $\operatorname{det} X^tX$ is based only on $x$, that is does not include any information of the bias and the confounder terms.
To make a comparison, isn the left panel the utility function is based on $(x, z/9)$, that is the  term modelling confounders is used. 
The control parameter $\nu$ in Algorithm~\ref{Wiens} is set equal to  $0.5$.
The grid used in the Algorithms for the $z$ is of hundred points in $-3/3$ and also hundred points were taken for the $x$ grid  in $-100:100$. We tried different discretization for the grids and the results were the same. 

In all our trials when comparing the designs obtained from a model, say $x$, and from a model with confounders, say $x+z/9$,  Algorithms~\ref{Drovandi} and~\ref{Wiens} 
give results more similar. 

\begin{figure}[h]  
\includegraphics[width=0.5\textwidth]{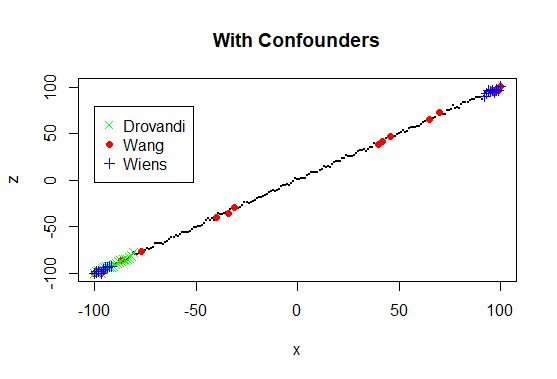} 
\includegraphics[width=0.5\textwidth]{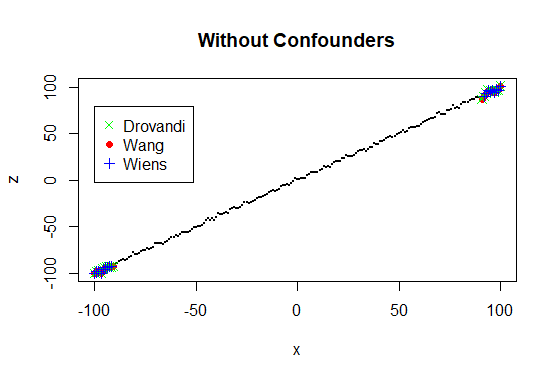} 
\caption{Different utilities for  Example~\ref{DrStWiens}}
\label{cc}
\end{figure}

} \end{example} 

%%%%%%%%%
%%%%%%%%%
%%%%%%%%%
%%%%%%%%%
%%%%%%%%% 

\section{Conclusions and future work}

In Section~\ref{section_moda} we reviewed literature which considers special cases of Model~(\ref{model}) in order to adapt ideas from classical model based optimal  DoE: \cite{montepiedra} and \cite{wiens} consider a linear model with a bias term while \cite{pescecladag} consider a linear model with confounders, searching for a design which minimises the mean square error of the LSE of the $\pmb{\theta}$ parameters. Here in Section~\ref{general_formulation} we follow that but guarding against the two different sources of bias. 

The algorithms in~\cite{drovandi} and \cite{stufken} offer methods of data selection from large datasets in a DoE context, however they do not guard against different sources of bias.
We are currently integrating the above ideas with those algorithms with the objective of providing efficient subsample selection methods for problems with known confounders and also with unknown confounders. 
%
%\subsection{General Equivalence Theorem}
%\label{get}
%We reviewed some results and algorithms for the selection of a subset from a large dataset 
%according to model-based principles in DoE and outlined their main features and the best conditions for their applicability. 
%For tall dataset the Algorithm in~\cite{drovandi} has proven to be efficient, but efficient algorithms for problems with known confounders and also with unknown confounders are still to be developed. 

We also presented some preliminary results on a unified theory to take into account selection bias, model bias and bias due to confounders 
in the choice of a subsample for an efficient estimation, in the least square sense, of parameters expressing the effect of interest. 
Still much work is needed to turn this into an algorithm for the selection of efficient subsamples from large or big data sets.
Furthermore the results in Section~\ref{general_formulation} need to be refined, possibly linking them with the algorithms in Section~\ref{section_moda}.
%
%More general model forms than the linear one could be analysed as well. Modularity in the model, possible symmetries in $\mathcal X$ and 
%could also be exploited to develop an efficient algorithm in the style of Algorithm~\ref{Wiens}. 
%A game theoretic approach to model based selection of a subsample from a large dataset is addressed in~\cite{pescecladag} and~\cite{waitewoods2018} but 
%it has not been reviewed here.

%\begin{acknowledgements}
%\end{acknowledgements}


\begin{thebibliography}{}
\bibitem{drovandi}
C.~C. Drovandi, C. Holmes, J.~M. McGree, K. Mengersen, S. Richardson and E.~G. Ryan, Principles of Experimental Design for Big Data Analysis, Statistical Science, 32(3), 385--404 (2017). 

\bibitem{dunson2018}
D.~B. Dunson, Statistics in the big data era: Failures of the machine, Statistics and Probability Letters, 136, 4--9 (2018).

\bibitem{FarawayAugustin2018}
J.~J. Faraway and N.~H. Augustin, When small data beats big data, Statistics and Probability Letters, 136, 142--145 (2018).

\bibitem{flassig2018model}
R.~J. Flassig and R. Schenkendorf, Model-based design of experiments: Where to go?, Manuscript (2018).

\bibitem{Harford2014}
T. Harford, Big data: are we making a big mistake?, Significance, December 2014, 14-19, reprint from The Financial Times (2014). 

\bibitem{Meng2018}
X.~L. Meng, Statistical paradises and paradoxes in big data (I): law of large populations, big data paradox, and the 2016 US presidential election, The Annals of Applied Statistics, 12(2), 685--726 (2018).

\bibitem{montepiedra}
G. Montepiedra and V.~V. Fedorov, Minimum bias design with constraints, Journal of Statistical Planning and Inference, 63, 97--111 (1997). 

\bibitem{pescecladag}
E. Pesce, E. Riccomagno and H.~P. Wynn, Passive and active observation: experimental design issues in big data, arXiv:1712.06916 (2017). 

\bibitem{scott2016bayes}
S.~L Scott, A.~W. Blocker, F.~V. Bonassi, H.~A. Chipman, Edward~I George, and Robert~E McCulloch, Bayes and big data: The consensus monte carlo algorithm, International Journal of Management Science and Engineering Management, 11(2), 78--88 (2016).

\bibitem{sharpes2018}
L.~D.. Sharpes, 
The role of statistics in the era of big data: Electronic health records for healthcare research,
Statistics and Probability Letters, 136, 105--110 (2018).

\bibitem{specialIssueStatProbLetters}
L.~M. Sangalli (editor), The role of Statistics in the era of Big Data,  Statistics and Probability Letters, 136 (Special issue). 

\bibitem{tibshirani2015statistical}
R. Tibshirani, M. Wainwright and T. Hastie, Statistical learning with sparsity: the lasso and generalizations, Chapman and Hall/CRC (2015).

\bibitem{stufken}
H. Wang, M. Yang and J. Stufken, Information-Based Optimal Subdata Selection for Big Data Linear Regression, Journal of the American Statistical Association (in press) (2018). 

\bibitem{wang2015statistical}
C. Wang, M.~H. Chen, E. Schifano, J. Wu and J. Yan, Statistical methods and computing for big data, Statistics and its interface, 9(4), 399 (2016).

%\bibitem{waitewoods2018}
%T. Waite and D. Woods, Minimax efficient random experimental design strategies with application to model-robust design for prediction,
%\newblock ?? CERCARE UN ARXIV O  PREPRINT

\bibitem{wiens}
D.P. Wiens, I-robust and D-robust designs on a finite design space, Statistics and Computing 28(2), 241--258 (2018). 

\end{thebibliography}


\end{document}